\documentclass[aps,prb,10pt,twocolumn,superscriptaddress]{revtex4-2}
\usepackage[ascii]{inputenc}
\usepackage{amsmath,amssymb,amsfonts,amsthm}
\usepackage{mathtools}
\usepackage{tikz}
\usetikzlibrary{decorations.pathreplacing,decorations.pathmorphing,3d,calc}
\usepackage{braket}
\usepackage{BOONDOX-uprscr}
\usepackage[caption=false]{subfig}
\usepackage[pdftex,bookmarks=false,colorlinks=true,linkcolor=blue,
citecolor=blue,filecolor=black,urlcolor=blue]{hyperref}

\DeclareMathOperator{\ud}{d}

\newcommand{\N}{\mathbb{N}}
\newcommand{\Z}{\mathbb{Z}}
\newcommand{\R}{\mathbb{R}}

\DeclarePairedDelimiter\abs{\lvert}{\rvert}
\DeclarePairedDelimiter\norm{\lVert}{\rVert}
\DeclarePairedDelimiter\comm{[}{]}
\DeclarePairedDelimiter\ceil{\lceil}{\rceil}

\DeclareMathOperator{\e}{e}
\DeclareMathOperator{\ad}{ad}

\graphicspath{{figures/}}
\makeatletter
\def\input@path{{figures/}}
\makeatother

\newtheorem{definition}{Definition}
\newtheorem{theorem}[definition]{Theorem}
\newtheorem{proposition}[definition]{Proposition}
\newtheorem{lemma}[definition]{Lemma}

\begin{document}

\title{Trotter error with commutator scaling for the Fermi-Hubbard model}

\author{Ansgar Schubert}
\email{ansgar.schubert@tum.de}
\affiliation{Technical University of Munich, CIT, Department of Computer Science, Boltzmannstra{\ss}e 3, 85748 Garching, Germany}
\author{Christian B.~Mendl}
\email{christian.mendl@tum.de}
\affiliation{Technical University of Munich, CIT, Department of Computer Science, Boltzmannstra{\ss}e 3, 85748 Garching, Germany}
\affiliation{Technical University of Munich, Institute for Advanced Study, Lichtenbergstra{\ss}e 2a, 85748 Garching, Germany}

\date{September 26, 2023}

\begin{abstract}
We derive higher-order error bounds with small prefactors for a general Trotter product formula, generalizing a result of Childs et al.~[Phys.~Rev.~X 11, 011020 (2021)].
We then apply these bounds to the real-time quantum time evolution operator governed by the Fermi-Hubbard Hamiltonian on one-dimensional and two-dimensional square and triangular lattices.
The main technical contribution of our work is a symbolic evaluation of nested commutators between hopping and interaction terms for a given lattice geometry.
The calculations result in explicit expressions for the error bounds in terms of the time step and Hamiltonian coefficients.
Comparison with the actual Trotter error (evaluated on a small system) indicates that the bounds still overestimate the error.
\end{abstract}

\maketitle

\section{Introduction}

This work is concerned with quantum simulation, i.e., approximating the time evolution operator $\e^{-i t H}$ of a quantum system governed by a time independent Hamiltonian $H$.
Splitting methods (also called product formulas) are a versatile and natural approach for this purpose, and can in principle be directly mapped to a quantum computer based on their circuit representation \cite{Wecker2015, Jiang2018, Childs2018, Smith2019, Heyl2019, Clinton2021, Bassman2021, Zhao2023}.
There is a rich mathematical theory and many studies on how to bound the simulation error when using product formulas \cite{Babbush2015, Jiang2018, Hadfield2018, Babbush2018, Childs2019, Tranter2019, Childs2021, Su2021, Campbell2021, McArdle2022, Layden2022, Chen2022, Myers2023}.
In particular, Childs et al.\ derived improved error bounds in a recent paper \cite{Childs2021}.
Their work applies to Hamiltonians $H = \sum_{\gamma=1}^{\Gamma} H_\gamma$ consisting of $\Gamma$ summands, where the time evolution governed by each individual summand can be evaluated exactly.
Using the same notation as Ref.~\cite{Childs2021}, a general (Trotter) product formula can be written as
\begin{equation}
\label{eq:general_product_formula}
\mathscr{S}_p(t) = \prod_{v=1}^{\Upsilon} \prod_{\gamma=1}^{\Gamma} \e^{-i t a_{v,\gamma} H_{\pi_v(\gamma)}}
\end{equation}
with $\Upsilon$ the number of stages, real coefficients $a_{v,\gamma}$ and $\pi_v(\gamma)$ the ordering of summands in a stage.
$p$ denotes the \emph{order} of the product formula, i.e.,
\begin{equation}
\label{eq:p_order_condition}
\mathscr{S}_p(t) = \e^{-i t H} + \mathcal{O}(t^{p+1}).
\end{equation}
Childs et al.\ derive the following bound:
\begin{equation} \label{eq:error-bound-alpha-comm-results}
\norm*{\mathscr{S}_p(t) -  \e^{-i t H}} = \mathcal{O}(t^{p+1}\tilde{\alpha}_{\text{comm}})
\end{equation}
where
\begin{equation} \label{eq:def-alpha-comm-results}
\tilde{\alpha}_{\text{comm}} = \sum_{\gamma_1, \dots, \gamma_{p+1} = 1}^{\Gamma} \norm*{ \comm{H_{\gamma_1}, \dots \comm{H_{\gamma_p}, H_{\gamma_{p+1}}}\dots } },
\end{equation}
$\norm{\cdot}$ denotes the spectral norm and $\comm{\cdot, \cdot}$ the commutator.
The authors also prove even tighter bounds for special cases, denoted error bounds with small prefactors.

In the present work, we generalize these results of Childs et al.\ by deriving general error bounds with small prefactors (see Theorem~\ref{thm:higher_order_bounds_prefactors} below), and evaluate commutators and bounds specifically for the Fermi-Hubbard model, representing a widely-studied and fundamental model class.
Compared to the conceptually similar studies~\cite{Su2021, Campbell2021}, we consider various lattice geometries.

Denoting the spin by $\sigma \in \{\uparrow, \downarrow \}$ and the underlying lattice by $\Lambda$, the Fermi-Hubbard Hamiltonian is defined as
\begin{equation}
\label{eq:def_H_FH}
H_{\text{FH}} = %
v \sum_{\langle i,j \rangle, \sigma} \left( a_{i\sigma}^\dagger a_{j\sigma}^{} + a_{j\sigma}^\dagger a_{i\sigma}^{} \right) + %
u \sum_{i \in \Lambda} n_{i\uparrow}^{} n_{i\downarrow}^{}
\end{equation}
where $i$ and $j$ refer to neighboring lattice sites in the first sum, $v \in \R$ is the kinetic hopping coefficient, and $u > 0$ the on-site interaction strength.
$a_{i\sigma}^\dagger$, $a_{i\sigma}^{}$ and $n_{i\sigma}^{}$ are the fermionic creation, annihilation and number operators, respectively, acting on site $i \in \Lambda$ and spin $\sigma \in \{ \uparrow, \downarrow \}$.

\section{General higher-order error bounds with small prefactors}

For notational simplicity, we express Eq.~\eqref{eq:general_product_formula} as
\begin{equation}
\label{eq:product_formula_sequential}
\mathscr{S}_p(t) = \e^{-i t A_K} \cdots \e^{-i t A_1} = \prod_{k=1}^K \e^{-i t A_k},
\end{equation}
where we have already merged compatible consecutive exponentials: $\e^{-i t a_1 H_{\gamma}} \e^{-i t a_2 H_{\gamma}} = \e^{-i t (a_1 + a_2) H_{\gamma}}$.

The multiplication order is relevant for non-commuting operators, and we fix the notation and convention:
\begin{equation}
\prod_{i=k}^n A_i = A_n A_{n-1} \cdots A_k
\end{equation}
and
\begin{equation}
\coprod_{i=n}^k A_i = A_k \cdots A_{n-1} A_n.
\end{equation}

The following theorem generalizes \cite[Appendix~M]{Childs2021}, with $\ad_A B = \comm{A, B}$ denoting the \emph{adjoint action}, and $\ad_A^q$ its $q$-fold application.

\begin{theorem}[Higher-order error bounds with small prefactors]
\label{thm:higher_order_bounds_prefactors}
Let $\mathscr{S}_p$ be a product formula of order $p$ in the representation \eqref{eq:product_formula_sequential}, and let $s \in \{ 1, \dots, K \}$.
Then
\begin{equation}
\label{eq:product_formula_error_bound}
\begin{split}
&\norm*{\mathscr{S}_p(t) - \e^{-i t H}} \\
&\le \frac{t^{p+1}}{(p+1)!} \Bigg( \sum_{j=2}^s \sum_{\substack{q_j + \dots + q_s = p \\ q_j \neq 0}} \binom{p}{q_j, \dots, q_s} \norm*{\ad_{A_s}^{q_s} \cdots \ad_{A_j}^{q_j} B_j} \\
&\: + \sum_{j=s+1}^K \sum_{\substack{q_{s+1} + \dots \\ + q_j = p \\ q_j \neq 0}} \binom{p}{q_{s+1}, \dots, q_j} \norm*{\ad_{A_{s+1}}^{q_{s+1}} \cdots \ad_{A_j}^{q_j} B_j} \Bigg)
\end{split}
\end{equation}
with
\begin{equation}
B_j = \sum_{\ell=1}^{j-1} A_{\ell}, \quad j = 2, \dots, K.
\end{equation}
\end{theorem}
Appendix~\ref{sec:proof_higher_order_bounds} contains a proof of the theorem.

As demonstration, we reproduce the results in \cite[Appendix~L]{Childs2021} using Eq.~\eqref{eq:product_formula_error_bound} applied to the Strang (second-order Suzuki) splitting rule
\begin{equation}
\mathscr{S}_2(t) = \e^{-i t H_1/2} \e^{-i t H_2} \e^{-i t H_1/2},
\end{equation}
with $H_1$ and $H_2$ Hermitian matrices and the overall Hamiltonian $H = H_1 + H_2$.
In the notation of Eq.~\eqref{eq:product_formula_sequential}, we have thus $K = 3$, $A_1 = \frac{1}{2} H_1$, $A_2 = H_2$, $A_3 = \frac{1}{2} H_1$ and $B_2 = \frac{1}{2} H_1$, $B_3 = \frac{1}{2} H_1 + H_2$.
We set $s = 2$.
Then the multinomial coefficients in \eqref{eq:product_formula_error_bound} evaluate to $1$, and
\begin{equation}
\label{eq:strang_bound_2h_demo}
\begin{split}
&\norm*{\mathscr{S}_2(t) - \e^{-i t H}} \\
&\le \frac{t^3}{3!} \left( \norm*{\ad_{A_2}^2 B_2} + \norm*{\ad_{A_3}^2 B_3} \right) \\
&= \frac{t^3}{6} \left( \frac{1}{2} \norm*{\comm{H_2, \comm{H_2, H_1}}} + \frac{1}{4} \norm*{\comm{H_1, \comm{H_1, H_2}}} \right) \\
&= \frac{t^3}{12} \norm*{\comm{H_2, \comm{H_2, H_1}}} + \frac{t^3}{24} \norm*{\comm{H_1, \comm{H_1, H_2}}},
\end{split}
\end{equation}
in agreement with \cite[Eq.~(L5)]{Childs2021}.

Evaluating Eq.~\eqref{eq:product_formula_error_bound} for the fourth-order Suzuki formula and two Hamiltonian terms likewise reproduces the coefficients in \cite[Eq.~(M13)]{Childs2021}, as expected.

We have empirically found that a centered $s = \ceil{\frac{K}{2}}$ generally leads to the tightest bounds (smallest coefficients).
For example, setting $s = 1$ instead of $2$ in the above Strang splitting demonstration results in the prefactor $\frac{t^3}{4}$ instead of $\frac{t^3}{12}$ in Eq.~\eqref{eq:strang_bound_2h_demo}.
Nevertheless, there are instances where some coefficients are slightly larger and others slightly smaller, e.g., for the fourth-order Suzuki formula and two Hamiltonian terms ($K = 11$), choosing $s = 6$ versus $s = 7$.
In these instances, the best choice for $s$ then depends on the actual norms of the nested commutators.
We will set $s = \ceil{\frac{K}{2}}$ throughout for the numerical calculations in Sect.~\ref{sec:results}.

It turns out that the commutator bound in Theorem~\ref{thm:higher_order_bounds_prefactors} can be further slightly improved for the special case of the Strang splitting method and more than two Hamiltonian terms, see \cite[Proposition~10]{Childs2021}.
To be self-contained, we cite the result here:
\begin{proposition}[Tight error bound for the second-order Suzuki formula, \cite{Childs2021}]
\label{prop:tight_bound_strang}
Let $H = \sum_{\gamma=1}^{\Gamma} H_\gamma$ be a Hamiltonian consisting of $\Gamma$ summands and $t \ge 0$.
Let $\mathscr{S}_2(t) = \coprod_{\gamma=\Gamma}^1 \e^{-i t H_\gamma/2} \prod_{\gamma=1}^{\Gamma} \e^{-i t H_\gamma/2}$ be the second-order Suzuki formula.
Then, the additive Trotter error can be bounded as
\begin{equation}
\label{eq:tight_bound_strang}
\begin{split}
&\norm*{\mathscr{S}_2(t) - \e^{-i t H}} \\
&\le \frac{t^3}{12} \sum_{\gamma_1=1}^{\Gamma} \norm*{\comm*{\sum_{\gamma_3=\gamma_1+1}^{\Gamma} H_{\gamma_3}, \comm*{\sum_{\gamma_2=\gamma_1+1}^{\Gamma} H_{\gamma_2}, H_{\gamma_1}}}} \\
&\quad + \frac{t^3}{24} \sum_{\gamma_1=1}^{\Gamma} \norm*{\comm*{H_{\gamma_1}, \comm*{\sum_{\gamma_2=\gamma_1+1}^{\Gamma} H_{\gamma_2}, H_{\gamma_1}}}}.
\end{split}
\end{equation}
\end{proposition}
The proof of this proposition in \cite{Childs2021} uses a telescoping sum and the self-similarity of the Strang splitting method for fewer terms.
This technique is not straightforwardly generalizable to higher-order product formulas.

We explicate the improvement offered by Proposition~\ref{prop:tight_bound_strang} in Appendix~\ref{sec:prefactor_comparison}.

\section{Lattice geometry and decomposition of the Hamiltonian}
\label{sec:lattice_decomposition}

Our goal is to decompose the Fermi-Hubbard Hamiltonian \eqref{eq:def_H_FH} into $H_{\text{FH}} = \sum_{\gamma=1}^{\Gamma} H_\gamma$, such that each $H_{\gamma}$ consists of local terms with disjoint support.
The operators $H_1, \dots, H_{\Gamma-1}$ will contain the kinetic terms and the last operator $H_{\Gamma}$ the on-site interactions.

Let $\Lambda$ denote the underlying spatial lattice of spin-endowed sites.
We will assume periodic boundary conditions throughout and finally take the thermodynamic limit of infinite lattice size.
For example, $\Lambda = \Z_{/L}$ for a one-dimensional lattice of size $L$ with periodic boundary conditions and even $L$.
The number of lattice sites is denoted $\abs{\Lambda}$.
It will turn out to be convenient to establish a sublattice $\Lambda' \subset \Lambda$ such that each $H_{\gamma}$ is translation invariant with respect to $\Lambda'$.
For the one-dimensional example $\Lambda = \Z_{/L}$, we will set $\Lambda' = (2 \Z)_{/L} = \{ 0, 2, \dots, L-2 \}$.

We define an individual hopping term as
\begin{equation}
\label{eq:def_hopping}
h_{ij\sigma}^{} = a_{i\sigma}^\dagger a_{j\sigma}^{} + a_{j\sigma}^\dagger a_{i\sigma}^{}
\end{equation}
(assuming $i \neq j$), and the number operator
\begin{equation}
n_{i\sigma}^{} = a_{i\sigma}^\dagger a_{i\sigma}^{},
\end{equation}
where $i, j \in \Lambda$ denote lattice sites and $\sigma \in \{ \uparrow, \downarrow \}$ the spin.
Additionally, we will encounter the ``signed hopping term''
\begin{equation}
\label{eq:def_signed_hopping}
\tilde{h}_{ij\sigma}^{} = a_{i\sigma}^\dagger a_{j\sigma}^{} - a_{j\sigma}^\dagger a_{i\sigma}^{},
\end{equation}
which is zero in case $i = j$.
Observe the (anti-)symmetry relations
\begin{subequations}
\begin{align}
h_{ij\sigma}^{}         &= h_{j i \sigma}^{} \label{eq:hopping-symmetry}, \\
\tilde{h}_{ij\sigma}^{} &= -\tilde{h}_{j i \sigma}^{}. \label{eq:signed_hopping-antisymmetry}
\end{align}
\end{subequations}

We now introduce common lattice geometries and corresponding Hamiltonian decompositions.
These define the scenarios studied in the present paper in detail.
The overarching theoretical framework is general and applicable to other geometries as well.

\subsection{One-dimensional lattice}

As mentioned, $\Lambda = \Z_{/L}$ with $L$ even and $\Lambda' = (2 \Z)_{/L}$.
We can decompose $H_{\text{FH}} = H_1 + H_2 + H_3$ with $H_1$ containing the ``even-odd'' kinetic hopping terms, $H_2$ the ``odd-even'' hopping terms and $H_3$ the interactions, as illustrated in Fig.~\ref{fig:decomposition_H_1d}.
Explicitly,
\begin{subequations}
\label{eq:H_terms_1d}
\begin{align}
H_1 &= v \sum_{i \in \Lambda'} \sum_{\sigma \in \{\uparrow, \downarrow\}} h_{i, i+1, \sigma}^{},\\
H_2 &= v \sum_{i \in \Lambda'} \sum_{\sigma \in \{\uparrow, \downarrow\}} h_{i-1, i, \sigma}^{},\\
H_3 &= u \sum_{i \in \Lambda'} \big( n_{i,\uparrow}^{} n_{i,\downarrow}^{} + n_{i+1,\uparrow}^{} n_{i+1,\downarrow}^{} \big).
\end{align}
\end{subequations}
By construction, each $H_{\gamma}$ is translation invariant with respect to $\Lambda'$, i.e., by a shift of two sites.

\begin{figure}[!ht]
\centering
\subfloat[hopping terms in $H_1$]{%
\begin{tikzpicture}[scale=0.9]
\draw[gray] (0, 0) -- (6, 0);
\foreach \x in {0, 2, 4}
{
    \draw[thick, red] (\x, 0) -- (\x+1, 0);
}
\foreach \x in {0,...,5}
{
    \draw[fill=white] (\x, 0) circle (0.1);
}
\end{tikzpicture}}\\
\subfloat[hopping terms in $H_2$]{%
\begin{tikzpicture}[scale=0.9]
\draw[gray] (0, 0) -- (6, 0);
\foreach \x in {1, 3, 5}
{
    \draw[thick, red] (\x, 0) -- (\x+1, 0);
}
\foreach \x in {0,...,5}
{
    \draw[fill=white] (\x, 0) circle (0.1);
}
\end{tikzpicture}}\\
\subfloat[interaction terms in $H_3$]{%
\begin{tikzpicture}[scale=0.9]
\draw[gray] (0, 0) -- (6, 0);
\foreach \x in {0,...,5}
{
    \fill[red]        (\x, 0) circle (0.14);
    \draw[fill=white] (\x, 0) circle (0.1);
}
\end{tikzpicture}}%
\caption{Decomposition of the Fermi-Hubbard Hamiltonian on a one-dimensional lattice as $H_{\text{FH}} = \sum_{\gamma=1}^3 H_{\gamma}$, such that terms within each $H_{\gamma}$ have disjoint support.}
\label{fig:decomposition_H_1d}
\end{figure}
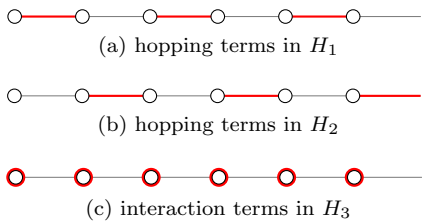

\subsection{Two-dimensional square lattice}

We consider a square $L \times L$ lattice with periodic boundary conditions and even $L$, i.e., $\Lambda = (\Z_{/L})^2$.

\begin{figure}[!ht]
\centering
\subfloat[hopping terms in $H_1$]{%
\begin{tikzpicture}[scale=0.9]
\draw[gray] (-2, -2) grid (2, 2);
\foreach \y in {-2, 0}
{
    \foreach \x in {-2, 0}
    {
        \draw[thick, red] (\x, \y) -- (\x+1, \y) -- (\x+1, \y+1) -- (\x, \y+1) -- cycle;
    }
}
\foreach \y in {-2,...,2}
{
    \foreach \x in {-2,...,2}
    {
        \draw[fill=white] (\x, \y) circle (0.1);
    }
}
\end{tikzpicture}}%
\hspace{0.05\columnwidth}%
\subfloat[hopping terms in $H_2$]{%
\begin{tikzpicture}[scale=0.9]
\draw[gray] (-2, -2) grid (2, 2);
\foreach \y in {-1, 1}
{
    \foreach \x in {-1, 1}
    {
        \draw[thick, red] (\x, \y) -- (\x+1, \y) -- (\x+1, \y+1) -- (\x, \y+1) -- cycle;
    }
}
\foreach \y in {-2,...,2}
{
    \foreach \x in {-2,...,2}
    {
        \draw[fill=white] (\x, \y) circle (0.1);
    }
}
\end{tikzpicture}} \\
\subfloat[interaction terms in $H_3$]{%
\begin{tikzpicture}[scale=0.9]
\draw[gray] (-2, -2) grid (2, 2);
\foreach \x in {-2,...,2}
{
    \foreach \y in {-2,...,2}
    {
        \fill[red]        (\x, \y) circle (0.14);
        \draw[fill=white] (\x, \y) circle (0.1);
    }
}
\end{tikzpicture}}
\caption{Decomposition of the Fermi-Hubbard Hamiltonian on a two-dimensional square lattice as $H_{\text{FH}} = \sum_{\gamma=1}^3 H_{\gamma}$, grouping the kinetic hopping terms into plaquettes.}
\label{fig:decomposition_H_2d_square}
\end{figure}
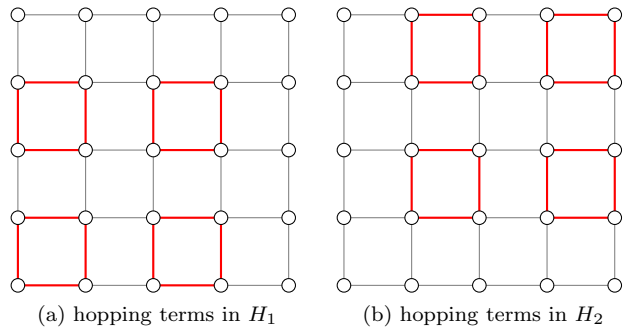

A small Trotter splitting error is achievable by partitioning the kinetic terms into plaquettes, as proposed and studied in Ref.~\cite{Campbell2021} and illustrated in Fig.~\ref{fig:decomposition_H_2d_square}.
By construction, each $H_\gamma$ is shift-invariant with respect to integer multiples of the vectors $(2, 0)$ and $(0, 2)$.
These span the sublattice $\Lambda' = ((2 \Z)_{/L})^2$.
The summands in the decomposition $H_{\text{FH}} = H_1 + H_2 + H_3$ can then be expressed as
\begin{subequations}
\label{eq:H_terms_2d_square}
\begin{align}
H_1 &= v \sum_{i \in \Lambda'} \sum_{\sigma \in \{\uparrow, \downarrow\}} \sum_{k=1}^4 h_{i+p_k, i+p_{k+1}, \sigma}^{},\\
H_2 &= v \sum_{i \in \Lambda'} \sum_{\sigma \in \{\uparrow, \downarrow\}} \sum_{k=1}^4 h_{i+p_k-(1,1), i+p_{k+1}-(1,1), \sigma}^{},\\
H_3 &= u \sum_{i \in \Lambda'} \sum_{k=1}^4 n_{i+p_k,\uparrow}^{} n_{i+p_k,\downarrow}^{}
\end{align}
\end{subequations}
where we have enumerated the vertex coordinates of a plaquette as $p_1 = (0, 0)$, $p_2 = (1, 0)$, $p_3 = (1, 1)$, $p_4 = (0, 1)$ and set $p_5 = p_1$.
This decomposition assumes that the four elementary hopping terms of a plaquette can be realized simultaneously.

Alternatively, we could separate the kinetic terms into horizontal and vertical hopping directions, and then each in turn into an even-odd partitioning, requiring four kinetic Hamiltonians in total.
As an advantage, each such Hamiltonian would only contain non-overlapping elementary hopping terms (instead of the four hopping terms subsumed in a plaquette).
However, it turns out that the resulting error bound is considerably weaker, and hence we opted for the plaquette Trotter splitting.

\subsection{Triangular lattice}

We consider the lattice structure shown in Fig.~\ref{fig:2d_triangular}.
By convention, the distance between nearest-neighbor lattice points is $1$.
\begin{figure}[!ht]
\centering
\begin{tikzpicture}[scale=0.9, >=stealth]
\draw[gray, ->] (0, 0) -- (3, 0) node[below] {$x$};
\draw[gray, ->] (0, 0) -- (0, 3) node[left]  {$y$};
\def\shifts{{{0, 0}}, {{1, 1}}, {{1, -1}}, {{-1, 1}}, {{-1, -1}}, {{0, 2}}, {{0, -2}}}
\foreach \sh in \shifts
{
    \begin{scope}[shift={(\sh[0]*1.5, \sh[1]*0.5*sqrt(3))}]
        \foreach \i in {0,...,5}
        {
            \draw[gray] (0, 0) -- ({cos(\i*60)}, {sin(\i*60)});
            \draw ({cos(\i*60)}, {sin(\i*60)}) -- ({cos((\i+1)*60)}, {sin((\i+1)*60)});
        }
    \end{scope}
}
\foreach \sh in \shifts
{
    \begin{scope}[shift={(\sh[0]*1.5, \sh[1]*0.5*sqrt(3))}]
        \draw[thick, red] (0, 0) -- (1, 0) -- ({cos(60)}, {sin(60)}) -- (0, 0);
    \end{scope}
}
\draw[->, very thick, blue] (0, 0) -- ( 0.925*1.5, {0.925*0.5*sqrt(3)});
\draw[->, very thick, blue] (0, 0) -- ( 0,         {0.925    *sqrt(3)});
\foreach \sh in \shifts
{
    \begin{scope}[shift={(\sh[0]*1.5, \sh[1]*0.5*sqrt(3))}]
        \draw[fill=black] (0, 0) circle (0.1);
        \foreach \i in {0,...,5}
        {
            \draw[fill=white] ({cos(\i*60)}, {sin(\i*60)}) circle (0.1);
        }
    \end{scope}
}
\end{tikzpicture}
\caption{Triangular lattice setup. The filled (black) lattice points are hexagon centers and define the sublattice $\Lambda'$. The red edges visualize the hopping terms of $H_1$ after decomposing the Hamiltonian as $H_{\text{FH}} = \sum_{\gamma=1}^4 H_{\gamma}$, see Eq.~\eqref{eq:H_terms_2d_triangular}.}
\label{fig:2d_triangular}
\end{figure}
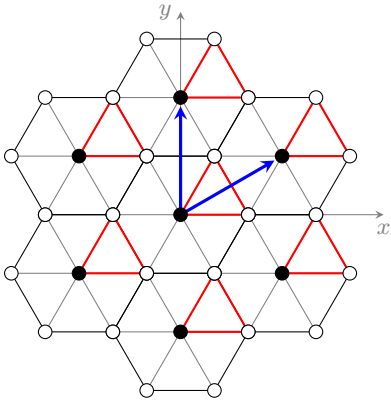
The triangular lattice $\Lambda$ is spanned by integer multiples of the unit cell vectors $(1, 0)$ and $(\frac{1}{2}, \frac{\sqrt{3}}{2})$.
We take the sublattice $\Lambda'$ to consist of the filled (black) points, i.e., the centers of the hexagons.
The unit cell vectors of $\Lambda'$ (blue vectors in Fig.~\ref{fig:2d_triangular}) have coordinates $(\frac{3}{2}, \frac{\sqrt{3}}{2})$ and $(0, \sqrt{3})$.

There are several possibilities of how to decompose the kinetic part of the Hamiltonian.
We choose to partition it into three summands $H_1, H_2, H_3$.
Each of them, in turn, contains (shifted copies of) \emph{three} hopping terms along the edges of a triangle, illustrated in red for $H_1$ in Fig.~\ref{fig:2d_triangular}.
$H_2$ is then obtained from $H_1$ via a rotation by $\frac{2\pi}{3}$ and $H_3$ via another rotation by $\frac{2\pi}{3}$.
The decomposition aims at a small number of Hamiltonian terms $H_{\gamma}$, such that the translated copies of the hopping operators contained in each $H_{\gamma}$ have disjoint support and hence commute.
Denoting the vertex coordinates of the hexagon at the origin as $g_k = (\cos((k-1) \frac{\pi}{3}), \sin((k-1) \frac{\pi}{3}))$, $k = 1, \dots, 6$, and setting $g_7 = g_1$, we can express
\begin{subequations}
\label{eq:H_terms_2d_triangular}
\begin{align}
\label{eq:H_terms_2d_triangular_kinetic}
H_{\ell} &= v \sum_{i \in \Lambda'} \sum_{\sigma \in \{\uparrow, \downarrow\}} \big( h_{i, i+g_{2\ell-1}, \sigma}^{} \\
&\qquad + h_{i+g_{2\ell-1}, i+g_{2\ell}, \sigma}^{} + h_{i+g_{2\ell}, i, \sigma}^{} \big) \text{ for } \ell = 1, 2, 3, \notag \\
\label{eq:H_terms_2d_triangular_interaction}
H_4 &= u \sum_{i \in \Lambda'} \left( n_{i,\uparrow}^{} n_{i,\downarrow}^{} + \frac{1}{3} \sum_{k=1}^6 n_{i+g_k,\uparrow}^{} n_{i+g_k,\downarrow}^{} \right).
\end{align}
\end{subequations}
The interaction part $H_4$ subsumes the local interactions on a hexagon.
The factor $\frac{1}{3}$ compensates for overcounting since each white lattice point is shared between three hexagons.
Note that the above representation of $H_4$ is invariant with respect to rotations by $\frac{\pi}{3}$.
Evaluating $\e^{-i t H_4}$ is possible based on the representation $H_4 = u \sum_{i \in \Lambda} n_{i\uparrow}^{} n_{i\downarrow}^{}$.

We have found it advantageous to use the following integer representation of lattice coordinates to avoid rounding errors:
We apply the linear conformal map $\R^2 \to \R^3$ defined in terms of lattice unit cell vectors by $(1, 0) \mapsto (2, -1, -1)$ and $(\frac{1}{2}, \frac{\sqrt{3}}{2}) \mapsto (1, 1, -2)$, which sends $\Lambda$ to a sublattice of the cubic lattice in three dimensions.
Each lattice point now has integer coordinates which sum to zero; these properties are inherited from the new unit cell vectors.
The unit cell vectors of $\Lambda'$ read $(3, 0, -3)$ and $(0, 3, -3)$.
Note that angles are preserved.

\section{Automated commutator evaluation and norm bounds}

Having a decomposition $H_{\text{FH}} = \sum_{\gamma=1}^{\Gamma} H_\gamma$ available, the next task consists of evaluating (nested) commutators between the $H_\gamma$ operators.

\subsection{Commutators of elementary operators}
\label{sec:commutators_elementary}

Commutators of the hopping and number operator terms follow from the fermionic anti-commutation relations, see, e.g., \cite{Helgaker2000}.
To be self-contained, we summarize them here, and provide a derivation in Appendix~\ref{sec:commutation_term_level}.
For sites $i, j, k, \ell \in \Lambda$ and spin orientations $\sigma, \tau \in \{ \uparrow, \downarrow \}$,
\begin{equation}
\label{eq:comm_disjoint_support}
\begin{split}
\comm{h_{ij\sigma}^{}, h_{k\ell\tau}^{}} &= 0, \quad %
\comm{\tilde{h}_{ij\sigma}^{}, \tilde{h}_{k\ell\tau}^{}} = 0, \quad %
\comm{h_{ij\sigma}^{}, \tilde{h}_{k\ell\tau}^{}} = 0, \\
\comm{h_{ij\sigma}^{}, n_{k\tau}^{}} &= 0, \quad %
\comm{\tilde{h}_{ij\sigma}^{}, n_{k\tau}^{}} = 0
\end{split}
\end{equation}
in case $\{ i, j \} \cap \{ k, \ell \} = \emptyset$ (disjoint support) or $\sigma \neq \tau$.

Number operators always commute: for all lattice sites $i, j \in \Lambda$ and $\sigma, \tau \in \{ \uparrow, \downarrow \}$,
\begin{equation}
\label{eq:comm_number_op}
\comm{n_{i\sigma}^{}, n_{j\tau}^{}} = 0.
\end{equation}

For $i, j, k \in \Lambda$ with $i \neq j$ and $j \neq k$ and $\sigma \in \{ \uparrow, \downarrow \}$,
\begin{subequations}
\begin{align}
\label{eq:comm_h_h}
\comm{h_{ij\sigma}^{}, h_{jk\sigma}^{}} &= \tilde{h}_{ik\sigma}^{}, \\
\label{eq:comm_g_g}
\comm{\tilde{h}_{ij\sigma}^{}, \tilde{h}_{jk\sigma}^{}} &= \tilde{h}_{ik\sigma}^{}, \\
\label{eq:comm_h_g}
\comm{h_{ij\sigma}^{}, \tilde{h}_{jk\sigma}^{}} &= \begin{cases} 2 \left(n_{i\sigma}^{} - n_{j\sigma}^{}\right), & i = k \\ h_{ik\sigma}^{}, & i \neq k \end{cases}
\end{align}
\end{subequations}
as well as, for $i \neq j$,
\begin{subequations}
\begin{align}
\label{eq:comm_h_n}
\comm{h_{ij\sigma}^{}, n_{j\sigma}^{}} &= \tilde{h}_{ij\sigma}^{}, \\
\label{eq:comm_g_n}
\comm{\tilde{h}_{ij\sigma}^{}, n_{j\sigma}^{}} &= h_{ij\sigma}^{}.
\end{align}
\end{subequations}

\subsection{Commutators of Hamiltonian operators}

In view of Theorem~\ref{thm:higher_order_bounds_prefactors}, we now evaluate (nested) commutators between the $H_\gamma$ operators.
The preceding subsection together with the general relations
\begin{equation}
\label{eq:commutator_sum}
\comm{A, B_1 + \dots + B_n} = \comm{A, B_1} + \dots + \comm{A, B_n}
\end{equation}
and
\begin{multline}
\label{eq:commutator_product}
\comm{A, B_1 \cdots B_n} = \comm{A, B_1} B_2 \cdots B_n \\
+ B_1 \comm{A, B_2} B_3 \cdots B_n + \dots \\
+ B_1 \cdots B_{n-1} \comm{A, B_n}
\end{multline}
show that these commutators can be expressed solely in terms of sums and products of the elementary operators $h_{ij\sigma}^{}$, $\tilde{h}_{ij\sigma}^{}$ and $n_{i\sigma}^{}$.

We can exploit translation invariance on $\Lambda'$ based on the following observation: let $A$ and $B$ be linear operators with local support (i.e., acting non-trivially only on a local region of lattice sites), and denote their versions after translation by $i \in \Lambda'$ as $A_i$ and $B_i$, respectively.
Then
\begin{equation}
\label{eq:commutator_sublattice_translations}
\comm*{\sum_{i \in \Lambda'} A_i, \sum_{k \in \Lambda'} B_k} = \sum_{i \in \Lambda'} \comm*{A_i, \sum_{\ell \in \Lambda'} B_{i + \ell}},
\end{equation}
where we have used the periodic boundary conditions of the overall lattice.
Together with the commutation relations for the individual hopping and interaction terms, one obtains, for example on the two-dimensional square lattice
\begin{equation}
\comm{H_1, H_2} = v^2 \sum_{i \in \Lambda'} \sum_{\sigma \in \{\uparrow, \downarrow\}} \left(\tilde{h}_{i, i+2e_1, \sigma}^{} - \tilde{h}_{i-e_1, i+e_1, \sigma}^{}\right)
\end{equation}
and
\begin{equation}
\comm{H_1, H_5} = v u \sum_{i \in \Lambda'} \sum_{\sigma \in \{\uparrow, \downarrow\}} \left(\tilde{h}_{i, i+e_1, \sigma}^{} \cdot (n_{i+e_1, \bar{\sigma}}^{} - n_{i, \bar{\sigma}}^{}) \right),
\end{equation}
where $\bar{\sigma}$ denotes the flipped spin.

\subsection{Automated symbolic commutator evaluation}

To automate the symbolic evaluation of (nested) commutators, we have implemented as Python package for this purpose, available at \cite{fermi_hubbard_commutators}.
It defines tailored Python classes to represent Hamiltonian operators and commutators between them.
The classes share an abstract base class \texttt{HamiltonianOp}:
\begin{itemize}
\item \texttt{HoppingOp} for representing $\alpha \cdot h_{ij\sigma}^{}$ with $\alpha \in \R$, storing the lattice coordinates $i, j$ as tuples, the spin $\sigma$ as integer, and a scalar coefficient $\alpha$ as floating point number,
\item \texttt{AntisymmHoppingOp} for representing $\alpha \cdot \tilde{h}_{ij\sigma}^{}$ with analogous member variables,
\item \texttt{NumberOp} to represent $\alpha \cdot n_{i\sigma}^{}$, storing the lattice coordinate $i$ as tuple, the spin $\sigma$ as integer and the scalar coefficient $\alpha$ as floating point number,
\item \texttt{ZeroOp} for the zero operation,
\item \texttt{ProductOp} to represent a product of Hamiltonian operators, stored in a Python list,
\item \texttt{SumOp} to represent a sum of Hamiltonian operators.
\end{itemize}
The $\texttt{ZeroOp}$ class could in principle be replaced by an empty sum, but we have found it convenient to indicate that an expression is zero and simplify derived expressions.
For example, in case a \texttt{ZeroOp} object appears as factor in a \texttt{ProductOp}, the overall product is zero.

In Sect.~\ref{sec:lattice_decomposition} we have consistently written each Hamiltonian term $H_{\gamma}$ as translated copies of some local term $h^{\text{loc}}_{\gamma}$ with respect to a sublattice $\Lambda'$, i.e., in the form
\begin{equation}
\label{eq:H_translated_local_op}
H_{\gamma} = \sum_{i \in \Lambda'} h^{\text{loc}}_{\gamma,i},
\end{equation}
with $h^{\text{loc}}_{\gamma,i}$ a copy of $h^{\text{loc}}_{\gamma}$ shifted by lattice vector $i$.
In our implementation, we mimic Eq.~\eqref{eq:H_translated_local_op} by only storing $h^{\text{loc}}_{\gamma}$ together with an instance of an auxiliary class \texttt{SubLattice} for representing $\Lambda'$.
This class contains the unit cell vectors of $\Lambda'$.

Evaluating commutators is achieved by a straightforward implementation of the relations in Sect.~\ref{sec:commutators_elementary} together with Eqs.~\eqref{eq:commutator_sum} and \eqref{eq:commutator_product}.
Translations are taken into account using Eq.~\eqref{eq:commutator_sublattice_translations}, which retains the form \eqref{eq:H_translated_local_op}.
For evaluating the commutator on the right in Eq.~\eqref{eq:commutator_sublattice_translations}, we enumerate all lattice vectors $\ell \in \Lambda'$ for which the support regions of $A_i$ and $B_{i + \ell}$ overlap.
Note that we only need to consider $i = 0$ (origin).

To apply Theorem~\ref{thm:higher_order_bounds_prefactors}, we have implemented a function to evaluate the expression on the right of Eq.~\eqref{eq:product_formula_error_bound}, retaining the nested commutators appearing in this expression in symbolic form at first.
In order to evaluate or upper-bound the spectral norm of a nested commutator like $\norm{\ad_{A_s}^{q_s} \cdots \ad_{A_j}^{q_j} B_j}$, we use the following strategy:
In case the operator acts non-trivially on at most 14 fermionic modes, where a mode refers to a lattice site and spin orientation, we compute the matrix representation of the operator and evaluate its spectral norm numerically exactly.
To exploit particle number conservation, which is adhered to by all involved operators, the computation uses the particle number sub-blocks of the matrix.
We also evaluate the exact norm for quadratic (free fermion) operators, i.e., consisting of linear combinations of hopping and number operators: in this case the spectral norm can be reduced to a sum of single-particle eigenvalues.
Otherwise, for non-quadratic operators involving more modes, we partition them into clusters supported on up to 14 modes each and make use of the triangle inequality to obtain an upper bound.
To avoid explicit dependence on system size, we report the spectral norm bounds as error per lattice site.

Another subtle point are telescoping effects.
Consider an operator containing two or more local summands, like $H_{\gamma} = \sum_{i \in \Lambda'} (A_i + B_i)$.
Then $\norm{H_{\gamma}} \le \sum_{i \in \Lambda'} \norm{A_i + B_i}$ by the triangle inequality.
Due to the periodic boundary conditions, we can also represent $H_{\gamma} = \sum_{i \in \Lambda'} (A_i + B_{i+\ell})$ for any fixed $\ell \in \Lambda'$, and correspondingly, $\norm{H_{\gamma}} \le \sum_{i \in \Lambda'} \norm{A_i + B_{i+\ell}}$.
The bound will depend on $\ell$ in general.
In order to arrive at a bound as tight as possible, we maximize the overlap (lattice support) of the local terms in our implementation.

\section{Commutator bounds and error analysis results}
\label{sec:results}

It is instructive to demonstrate the analytic evaluation of the commutator bounds explicitly for the concrete example of a one-dimensional lattice and the Strang (second-order Suzuki) splitting method.
For higher-order product rules and two-dimensional lattices, we will use the automated symbolic evaluation since the algebraic manipulations become rather tedious.

\subsection{Analytical derivation for a one-dimensional lattice and Strang splitting}

We consider the Hamiltonian terms in Eq.~\eqref{eq:H_terms_1d} and the Strang (second-order Suzuki) formula.
We make use of the theoretical bound of Proposition~\ref{prop:tight_bound_strang}, as concretized in Eq.~\eqref{eq:tight_bound_strang_three_terms} for the present setting.
Regarding the individual commutators, we first evaluate the commutator of the kinetic Hamiltonian terms:
\begin{equation}
\comm*{H_2, H_1} = v^2 \sum_{i \in \Lambda'} \sum_{\sigma \in \{\uparrow, \downarrow\}} \left(\tilde{h}_{i-1, i+1, \sigma}^{} - \tilde{h}_{i, i+2, \sigma}^{}\right).
\end{equation}
The nested commutator with $H_1$, $H_2$ and $H_3$ is then
\begin{subequations}
\begin{align}
\label{eq:comm_bound_1d_H1_H2_H1}
\comm*{H_1, \comm*{H_2, H_1}} &= 2 v^3 \sum_{i \in \Lambda'} \sum_{\sigma \in \{\uparrow, \downarrow\}} \left(h_{i-2, i+1, \sigma}^{} - h_{i-1, i, \sigma}^{}\right),\\
\comm*{H_2, \comm*{H_2, H_1}} &= 2 v^3 \sum_{i \in \Lambda'} \sum_{\sigma \in \{\uparrow, \downarrow\}} \left(h_{i, i+1, \sigma}^{} - h_{i-1, i+2, \sigma}^{}\right)
\end{align}
\end{subequations}
and
\begin{equation}
\begin{split}
&\comm*{H_3, \comm*{H_2, H_1}} \\
&= v^2 u \sum_{i \in \Lambda'} \sum_{\sigma \in \{\uparrow, \downarrow\}} \Big(h_{i-1, i+1, \sigma}^{} \cdot \left(n_{i-1, \bar{\sigma}}^{} - n_{i+1, \bar{\sigma}}^{}\right) \\
&\hspace{29mm} + h_{i, i+2, \sigma}^{} \cdot \left(n_{i+2, \bar{\sigma}}^{} - n_{i, \bar{\sigma}}^{}\right)\Big),
\end{split}
\end{equation}
with $\bar{\sigma}$ denoting the flipped spin.
Next, we evaluate the commutator between a kinetic and interaction Hamiltonian term:
\begin{subequations}
\begin{align}
\comm*{H_3, H_1} &= v u \sum_{i \in \Lambda'} \sum_{\sigma \in \{\uparrow, \downarrow\}} \tilde{h}_{i, i+1, \sigma}^{} \cdot \left(n_{i, \bar{\sigma}}^{} - n_{i+1, \bar{\sigma}}^{}\right),\\
\comm*{H_3, H_2} &= v u \sum_{i \in \Lambda'} \sum_{\sigma \in \{\uparrow, \downarrow\}} \tilde{h}_{i-1, i, \sigma}^{} \cdot \left(n_{i-1, \bar{\sigma}}^{} - n_{i, \bar{\sigma}}^{}\right).
\end{align}
\end{subequations}
Computing nested commutators then leads to
\begin{multline}
\comm*{H_1, \comm*{H_3, H_1}} \\
= - 4 v^2 u \sum_{i \in \Lambda'} \Big( \left(n_{i, \uparrow}^{} - n_{i+1, \uparrow}^{}\right) \left(n_{i, \downarrow}^{} - n_{i+1, \downarrow}^{}\right) \\
+ \tilde{h}_{i, i+1, \uparrow}^{} \cdot \tilde{h}_{i, i+1, \downarrow}^{} \Big),
\end{multline}
\begin{equation}
\label{eq:comm_H2_H3_H1}
\begin{split}
&\comm*{H_2, \comm*{H_3, H_1}} \\
&= v^2 u \sum_{i \in \Lambda'} \sum_{\sigma \in \{\uparrow, \downarrow\}} \Big( \left(h_{i-1, i+1, \sigma}^{} - h_{i, i+2, \sigma}^{}\right) \\
&\hspace{34mm}\cdot \left(n_{i, \bar{\sigma}}^{} - n_{i+1, \bar{\sigma}}^{}\right) \\
&\hspace{10mm}+ \tilde{h}_{i-1, i, \sigma}^{} \cdot \tilde{h}_{i, i+1, \bar{\sigma}}^{} + \tilde{h}_{i, i+1, \sigma}^{} \cdot \tilde{h}_{i+1, i+2, \bar{\sigma}}^{} \Big),
\end{split}
\end{equation}
and
\begin{equation}
\begin{split}
&\comm*{H_3, \comm*{H_3, H_1}} = v u^2 \sum_{i \in \Lambda'} \sum_{\sigma \in \{\uparrow, \downarrow\}} h_{i, i+1, \sigma}^{} \left(n_{i, \bar{\sigma}}^{} - n_{i+1, \bar{\sigma}}^{}\right)^2.
\end{split}
\end{equation}
Analogous expressions hold for $\comm{H_1, \comm{H_3, H_2}}$, $\comm{H_2, \comm{H_3, H_2}}$ and $\comm{H_3, \comm{H_3, H_2}}$.

We report an upper bound on the spectral norm of the commutators \emph{per lattice site}, by omitting the summation over $i \in \Lambda'$ and diving by $2$ (since $\Lambda'$ only covers every second site).
For example, applying this procedure to the expression in Eq.~\eqref{eq:comm_bound_1d_H1_H2_H1} and using the triangle inequality gives
\begin{equation}
\begin{split}
\frac{1}{\abs{\Lambda}} \norm*{\comm*{H_1, \comm*{H_2, H_1}}} %
&\le \abs{v}^3 \sum_{\sigma \in \{\uparrow, \downarrow\}} \norm*{h_{-2, 1, \sigma}^{} - h_{-1, 0, \sigma}^{}} \\
&= 4 \abs{v}^3.
\end{split}
\end{equation}
One also recognizes that $\norm{n_{i, \sigma}^{} - n_{j, \sigma}^{}} = 1$ for any $i \neq j$.

Regarding $\norm{\comm{H_2, \comm{H_3, H_1}}}$, we form the matrix representation $\in \R^{256 \times 256}$ of the summand in Eq.~\eqref{eq:comm_H2_H3_H1} (for fixed $i \in \Lambda'$) and compute its exact spectral norm.
Specifically, the matrix is symmetric and has largest eigenvalue and singular value $4$.
For comparison, bounding the norm using the triangle and sub-multiplicative properties yields $8$, which is thus not tight.

Inserting everything into Eq.~\eqref{eq:tight_bound_strang_three_terms} leads to
\begin{equation}
\label{eq:bound_eval_1d_strang}
\frac{1}{\abs{\Lambda}} \norm*{\mathscr{S}_2(t) - \e^{-i t H_{\text{FH}}}} \le \frac{t^3}{6} \left( 3 \abs{v}^3 + 4 \abs{v}^2 \abs{u} + \abs{v} \abs{u}^2 \right).
\end{equation}
This formula is the first concluding result of this paper.

\subsection{Higher-order splitting methods for a one-dimensional lattice}

We now make use of the automated symbolic commutator evaluation to derive error bounds on fourth-order methods as well, and compare with the empirical error evaluated on a small system.
The empirical error refers to the numerically exact evaluation of the time evolution operator, splitting method and deviation between them.
We first consider the fourth-order Suzuki formula.
In general, higher-order Suzuki formulas can be defined recursively via \cite{Suzuki1991}
\begin{align}
\mathscr{S}_2(t) &= \e^{-i t H_1/2} \cdots \e^{-i t H_\Gamma/2} \e^{-i t H_\Gamma/2} \cdots \e^{-i t H_1/2}, \\
\mathscr{S}_{2k}(t) &= \mathscr{S}_{2k-2}^2 (u_k t) \, \mathscr{S}_{2k-2} ((1 - 4 u_k) t) \, \mathscr{S}_{2k-2}^2 (u_k t)
\end{align}
with $u_k = 1 / (4 - 4^{1/(2k-1)})$ and $k \in \N$, $k \ge 2$.
The fourth-order Suzuki formula with three Hamiltonian terms has been analyzed in \cite[Appendix~M and Proposition~M.2]{Childs2021}, and our proof of Theorem~\ref{thm:higher_order_bounds_prefactors} generalizes the technique there.
Programmatically evaluating the coefficients in Eq.~\eqref{eq:product_formula_error_bound} reproduces \cite[Table~II]{Childs2021} when setting $s = 10$.
It turns out that $s = \ceil{\frac{K}{2}} = 11$ leads to a sharper bound, since some coefficients are smaller; for example, the coefficient in front of $\norm{\comm{H_3, \comm{H_3, \comm{H_3, \comm{H_3, H_2}}}}}$ is $0.0628$ for $s = 10$ and $0.0316$ for $s = 11$.

Together with evaluating norm bounds of the nested commutators, one arrives at
\begin{equation}
\label{eq:bound_eval_1d_suzuki4}
\begin{split}
&\frac{1}{\abs{\Lambda}} \norm*{\mathscr{S}_4(t) - \e^{-i t H_{\text{FH}}}} \le t^5 \big( 1.3405 \abs{v}^5 + 8.8233 \abs{v}^4 \abs{u} \\
&\quad + 2.3945 \abs{v}^3 \abs{u}^2 + 0.4137 \abs{v}^2 \abs{u}^3 + 0.06001 \abs{v} \abs{u}^4 \big)
\end{split}
\end{equation}
for the fourth-order Suzuki method and one-dimensional lattice setting.

For comparison, we investigate another fourth-order splitting method: the symmetric scheme AK~11-4 for three terms by Auzinger et al.~\cite{Auzinger2017}.
This leads to the estimate
\begin{equation}
\label{eq:bound_eval_1d_ak11_4}
\begin{split}
&\frac{1}{\abs{\Lambda}} \norm*{\mathscr{S}_4(t) - \e^{-i t H_{\text{FH}}}} \le t^5 \big( 3.0745 \abs{v}^5 + 28.2247 \abs{v}^4 \abs{u} \\
&\quad + 13.4738 \abs{v}^3 \abs{u}^2 + 4.9908 \abs{v}^2 \abs{u}^3 + 0.9155 \abs{v} \abs{u}^4 \big)
\end{split}
\end{equation}
for the AK~11-4 method, which is noticeably larger than the bound for the Suzuki method in Eq.~\eqref{eq:bound_eval_1d_suzuki4}.

As a remark, the norms of all nested commutators contained in Eqs.~\eqref{eq:bound_eval_1d_suzuki4} and \eqref{eq:bound_eval_1d_ak11_4} have been computed numerically exactly.

\begin{figure}[!ht]
\centering
\includegraphics[width=\columnwidth]{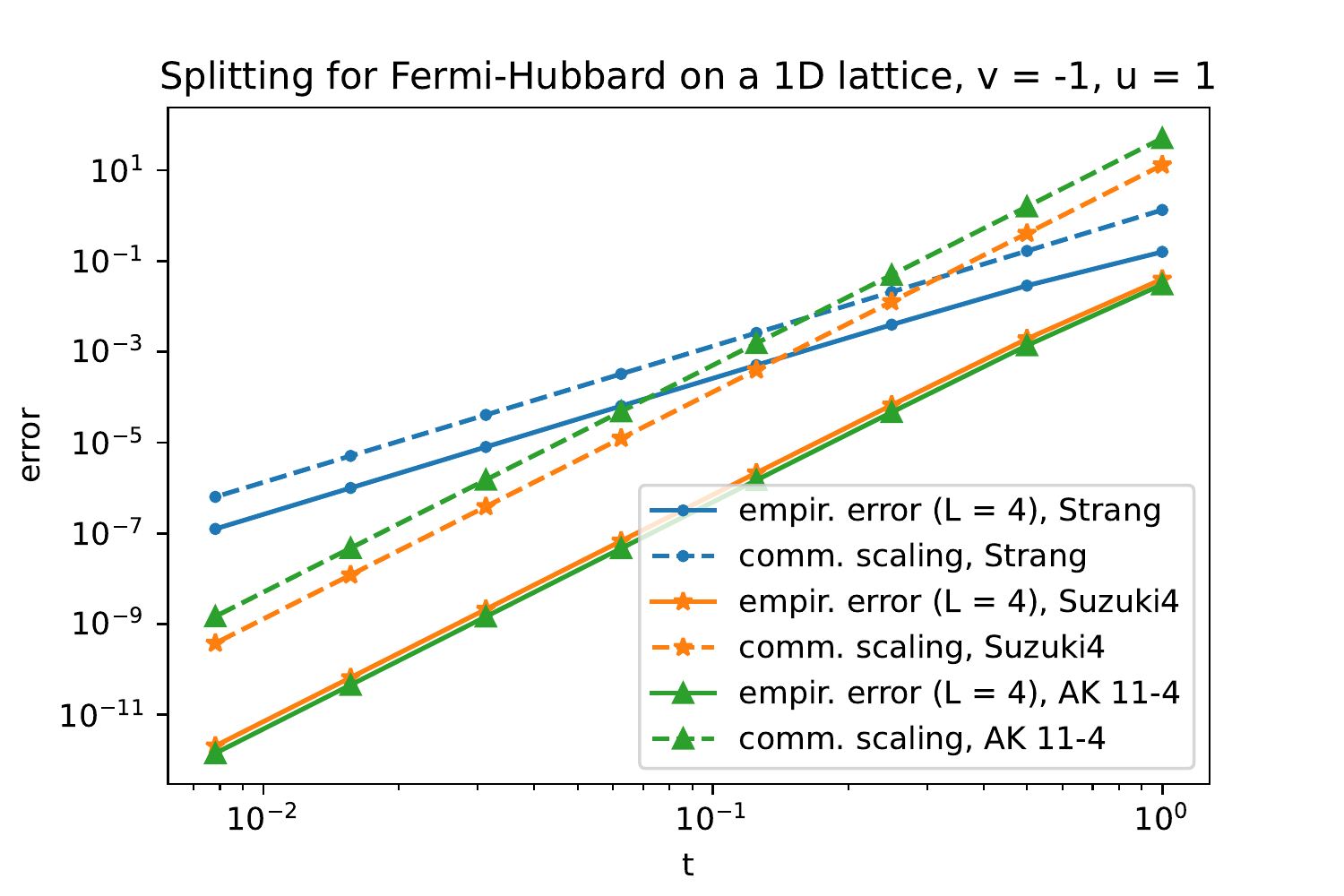}
\caption{Trotter splitting error $\frac{1}{\abs{\Lambda}} \norm*{\mathscr{S}(t) - \e^{-i t H_{\text{FH}}}}$ for the Fermi-Hubbard model on a one-dimensional lattice, comparing commutator bounds with the empirical error ($L = 4$ sites). The Hamiltonian coefficients are set to $v = -1$ and $u = 1$.}
\label{fig:fh_comm_1d_error}
\end{figure}

Fig.~\ref{fig:fh_comm_1d_error} summarizes the commutator bounds and empirical error on a one-dimensional lattice, both for the Strang (second-order Suzuki) formula and the fourth-order methods.
The commutator bounds overestimate the empirical error by an order of magnitude for Strang splitting, and by around three orders of magnitude for the fourth-order methods.
We evaluate the empirical error only for a small system due to the steep increase in computational cost for larger $L$, and thus the error for larger systems could deviate from the values in the plot.
In any case, the results indicate that the commutator bounds are far from tight.
Interestingly, the empirical error using the AK~11-4 scheme is slightly smaller compared to the Suzuki method, reversed from the bounds in Eqs.~\eqref{eq:bound_eval_1d_suzuki4} and \eqref{eq:bound_eval_1d_ak11_4}.

\subsection{Two-dimensional square lattice}

We now apply splitting methods to the Fermi-Hubbard model on a two-dimensional square lattice, using the three Hamiltonian terms in Eq.~\eqref{eq:H_terms_2d_square}.

For Strang splitting, we obtain the bound
\begin{multline}
\frac{1}{\abs{\Lambda}} \norm*{\mathscr{S}_2(t) - \e^{-i t H_{\text{FH}}}} \\
\le \frac{t^3}{6} \left( 4.4142 \abs{v}^3 + 8.0889 \abs{v}^2 \abs{u} + 1.3062 \abs{v} \abs{u}^2 \right)
\end{multline}
on the error per lattice site.
We have already accounted for a factor $4$ due to the sublattice $\Lambda'$ covering one quarter of all sites.

For the fourth-order Suzuki formula, an analogous calculation leads to
\begin{equation}
\label{eq:bound_eval_2d_square_suzuki4}
\begin{split}
&\frac{1}{\abs{\Lambda}} \norm*{\mathscr{S}_4(t) - \e^{-i t H_{\text{FH}}}} \le t^5 \big( 2.1485 \abs{v}^5 + 92.1642 \abs{v}^4 \abs{u} \\
&\quad + 14.3445 \abs{v}^3 \abs{u}^2 + 1.0712 \abs{v}^2 \abs{u}^3 + 0.07938 \abs{v} \abs{u}^4 \big).
\end{split}
\end{equation}
One notices the relatively small last coefficient.

\subsection{Two-dimensional triangular lattice}

Finally, we consider the triangular lattice shown in Fig.~\ref{fig:2d_triangular} and the four Hamiltonian terms in Eq.~\eqref{eq:H_terms_2d_triangular}.

For Strang splitting, we obtain the bound
\begin{multline}
\label{eq:bound_eval_2d_triangular_strang}
\frac{1}{\abs{\Lambda}} \norm*{\mathscr{S}_2(t) - \e^{-i t H_{\text{FH}}}} \\
\le \frac{t^3}{6} \left( 39.4721 \abs{v}^3 + 20.1594 \abs{v}^2 \abs{u} + 1.9546 \abs{v} \abs{u}^2 \right)
\end{multline}
on the error per lattice site.
The lattice $\Lambda$ outnumbers $\Lambda'$ by a factor $3$, which we have already taken into account here.

The analogous error bound on the fourth-order Suzuki formula reads
\begin{equation}
\label{eq:bound_eval_2d_triangular_suzuki4}
\begin{split}
&\frac{1}{\abs{\Lambda}} \norm*{\mathscr{S}_4(t) - \e^{-i t H_{\text{FH}}}} \le t^5 \big( 124.815 \abs{v}^5 + 493.917 \abs{v}^4 \abs{u} \\
&\quad + 60.4106 \abs{v}^3 \abs{u}^2 + 2.9855 \abs{v}^2 \abs{u}^3 + 0.1206 \abs{v} \abs{u}^4 \big).
\end{split}
\end{equation}
The coefficients are considerably larger compared to the analogous bound for the square lattice in Eq.~\eqref{eq:bound_eval_2d_square_suzuki4}.
As indication for the origin of this deviation, we remark that the number of substeps $K$ in the splitting methods differs: $K = 21$ for the square lattice (three Hamiltonian operators in Eqs.~\eqref{eq:H_terms_2d_square}) compared to $K = 31$ for the triangular lattice (four Hamiltonian operators in Eqs.~\eqref{eq:H_terms_2d_triangular}).

\section{Conclusions and outlook}

The comparison with the empirical Trotter error in Fig.~\ref{fig:fh_comm_1d_error} indicates that the commutator scaling bounds are not perfectly tight.
Interestingly, in the proof of Theorem~\ref{thm:higher_order_bounds_prefactors}, the expression \eqref{eq:S_diff_expitH} together the formulas for $\mathcal{C}_j(t)$ still describe the exact difference between the Trotter formula and actual time evolution operator.
The exponential rotations contained in the integrand of \eqref{eq:S_diff_expitH} lead to cancellations which are not accounted for when applying the triangle inequality.
In a future work, one could potentially exploit this observation to arrive at sharper bounds.

The decomposition of a Hamiltonian as $H = \sum_{\gamma} H_\gamma$ is not unique and can affect the resulting Trotter error.
The bounds developed in this work could be used to guide the search for a decomposition with smallest possible error.

We have studied three lattice geometries in detail as representative examples.
The generality of the framework and of the implementation in \cite{fermi_hubbard_commutators} allow for a straightforward application to other geometries and decompositions as well.
Moreover, an extension by additional particle species (like bosonic particles) or pairing interactions for superconductivity, for example, is likewise conceivable.

Ref.~\cite{Campbell2021} considers a variant of the Fermi-Hubbard model where the number operators in the interaction part are shifted by $\frac{1}{2}$.
This convention leads to concise expressions for the nested commutators and has the effect that some of them have smaller norms.
We leave an exploration of this approach and a detailed comparison with Ref.~\cite{Campbell2021} for future work.
Note that the commutators from Sect.~\ref{sec:commutators_elementary} remain unaffected when shifting the number operators by a constant.

Finally, we would like to remark that Trotter splitting methods can in principle be further improved by numerically optimizing the individual substeps tailored for a given Hamiltonian \cite{Mansuroglu2023, Tepaske2023, McKeever2023, Kotil2022}.

\begin{acknowledgments}
This research is part of the Munich Quantum Valley, which is supported by the Bavarian state government with funds from the Hightech Agenda Bayern Plus.
\end{acknowledgments}

\bibliography{references}

\appendix

\newpage

\section{Proof of the higher-order error bounds}
\label{sec:proof_higher_order_bounds}

This section contains a proof of Theorem~\ref{thm:higher_order_bounds_prefactors}.

We generalize the derivation in \cite[Appendix~M]{Childs2021}:
\begin{equation}
\begin{split}
&\frac{\ud}{\ud t} \mathscr{S}_p(t) - (-i H) \mathscr{S}_p(t) \\
&= \comm*{\e^{-i t A_K}, -i A_{K-1}} \e^{-i t A_{K-1}} \cdots \e^{-i t A_1} \\
&\: + \comm*{\e^{-i t A_K} \e^{-i t A_{K-1}}, -i A_{K-2}} \e^{-i t A_{K-2}} \cdots \e^{-i t A_1} \\
&\: + \dots \\
&\: + \comm*{\e^{-i t A_K} \cdots \e^{-i t A_2}, -i A_1} \e^{-i t A_1} \\
&= \comm*{\e^{-i t A_K}, -i A_{K-1}} \e^{-i t A_{K-1}} \cdots \e^{-i t A_1} \\
&\: + \e^{-i t A_K} \comm*{\e^{-i t A_{K-1}}, -i A_{K-2}} \e^{-i t A_{K-2}} \cdots \e^{-i t A_1} \\
&\: + \comm*{\e^{-i t A_K}, -i A_{K-2}} \e^{-i t A_{K-1}} \e^{-i t A_{K-2}} \cdots \e^{-i t A_1} \\
&\: + \dots \\
&\: + \e^{-i t A_K} \cdots \e^{-i t A_4} \e^{-i t A_3} \comm*{\e^{-i t A_2}, -i A_1} \e^{-i t A_1} \\
&\: + \e^{-i t A_K} \cdots \e^{-i t A_4} \comm*{\e^{-i t A_3}, -i A_1} \e^{-i t A_2} \e^{-i t A_1} \\
&\: + \dots \\
&\: + \comm*{\e^{-i t A_K}, -i A_1} \e^{-i t A_{K-1}} \cdots \e^{-i t A_2} \e^{-i t A_1}.
\end{split}
\end{equation}
For the first equal sign, we have used the consistency condition of the integration method, $\sum_{k=1}^K A_k = H$, and for the second equal sign, we have sequentially expanded the commutators, according to the blueprint
\begin{multline}
\comm{A_1 A_2 A_3, B} \\
= A_1 A_2 \comm{A_3, B} + A_1 \comm{A_2, B} A_3 + \comm{A_1, B} A_2 A_3.
\end{multline}
As next step, we sum up all commutators appearing at the same position within the chain of matrix exponentials, which leads to
\begin{equation}
\begin{split}
&\frac{\ud}{\ud t} \mathscr{S}_p(t) - (-i H) \mathscr{S}_p(t) \\
&= \sum_{j=2}^K \prod_{k=j+1}^K \e^{-i t A_k} \comm*{\e^{-i t A_j}, -i \sum_{\ell=1}^{j-1} A_{\ell}} \prod_{k'=1}^{j-1} \e^{-i t A_{k'}}.
\end{split}
\end{equation}
As in \cite{Childs2021}, we now factor out matrix exponentials on the left and right side; for that purpose, we fix some index $s \in \{ 1, \dots, K \}$ and introduce
\begin{equation}
\mathscr{S}_{\text{left}}(t) = \prod_{k=s+1}^K \e^{-i t A_k}, \quad
\mathscr{S}_{\text{right}}(t) = \prod_{k=1}^s \e^{-i t A_k}.
\end{equation}
This leads to
\begin{equation}
\label{eq:S_diff_left_right}
\frac{\ud}{\ud t} \mathscr{S}_p(t) - (-i H) \mathscr{S}_p(t) = \mathscr{S}_{\text{left}}(t) \, \mathcal{T}(t) \, \mathscr{S}_{\text{right}}(t)
\end{equation}
with
\begin{equation}
\label{eq:T_def}
\mathcal{T}(t) = \sum_{j=2}^K \mathcal{T}_{\text{left},j}(t) \comm*{\e^{-i t A_j}, -i \sum_{\ell=1}^{j-1} A_{\ell}} \mathcal{T}_{\text{right},j}(t),
\end{equation}
where
\begin{subequations}
\begin{align}
\mathcal{T}_{\text{left},j}(t) &= \begin{cases} \coprod_{k=j}^{s+1} \e^{i t A_k} & j > s \\ 1 & j = s \\ \prod_{k=j+1}^s \e^{-i t A_k} & j < s \end{cases} \\
\mathcal{T}_{\text{right},j}(t) &= \begin{cases} \prod_{k=s+1}^{j-1} \e^{-i t A_k} & j - 1 > s \\ 1 & j - 1 = s \\ \coprod_{k=s}^j \e^{i t A_k} & j - 1 < s \end{cases}
\end{align}
\end{subequations}
Further following the derivation in \cite{Childs2021}, we express the commutator in Eq.~\eqref{eq:T_def} via
\begin{equation}
\begin{split}
\comm*{ \e^{t X}, Y } %
&= \e^{t X} \int_0^t \ud\tau \e^{-\tau X} \comm{X, Y} \e^{\tau X} \\
&= \int_0^t \ud\tau \e^{\tau X} \comm{X, Y} \e^{-\tau X} \e^{t X},
\end{split}
\end{equation}
using the first variant in case $j > s$, and the second variant in case $j \le s$.
This leads to
\begin{multline}
\label{eq:T_int_repr}
\mathcal{T}(t) \\
= \sum_{j=2}^K \mathcal{T}_j(t) \left( \int_0^t \ud\tau \e^{\pm i \tau A_j} \comm*{-i A_j, -i B_j} \e^{\mp i \tau A_j} \right) \mathcal{T}_j^{\dagger}(t)
\end{multline}
with the first (upper) sign corresponding to $j > s$ and the lower sign to $j \le s$,
\begin{equation}
B_j = \sum_{\ell=1}^{j-1} A_{\ell}, \quad j = 2, \dots, K,
\end{equation}
and
\begin{equation}
\mathcal{T}_j(t) = \begin{cases} \coprod_{k=j-1}^{s+1} \e^{i t A_k} & j > s + 1 \\ 1 & j = s, s + 1 \\ \prod_{k=j+1}^s \e^{-i t A_k} & j < s \end{cases}
\end{equation}
Together with the Lie-algebraic identity $\e^A B \e^{-A} = \e^{\ad_A} B$ for $\ad_A B = \comm{A, B}$, we can express \eqref{eq:T_int_repr} as
\begin{equation}
\label{eq:T_ad_repr}
\mathcal{T}(t) %
= \sum_{j=2}^K \mathcal{T}_j(t) \left( \int_0^t \ud\tau \e^{\pm i \tau \ad_{A_j}} \ad_{i A_j} (i B_j) \right) \mathcal{T}_j^{\dagger}(t).
\end{equation}
Next, we use \cite[Theorem~5]{Childs2021}, which in turn is based on a repeated application of the Taylor series expansion (for integer $q \in \N_{\ge 1}$):
\begin{multline}
\e^{\tau \ad_A} B%
= B + \tau \ad_A B + \dots + \frac{\tau^{q-1}}{(q-1)!} \ad_A^{q-1} B \\
+ \int_0^{\tau} \ud\tau_2 \, \frac{\tau_2^{q-1}}{(q-1)!} \e^{(\tau - \tau_2) \ad_A} \ad_A^q B.
\end{multline}
First applying this theorem to $\e^{\pm i \tau \ad_{A_j}}$ inside the integral of \eqref{eq:T_ad_repr} and then to the conjugations by $\mathcal{T}_j(t)$ facilitates the series expansion
\begin{multline}
\label{eq:abstract_order_series_expansion}
\mathcal{T}_j(t) \left( \int_0^t \ud\tau \e^{\pm i \tau \ad_{A_j}} \ad_{i A_j} (i B_j) \right) \mathcal{T}_j^{\dagger}(t) \\
= C_{j,0} + C_{j,1} t + \dots + C_{j,p-1} t^{p-1} + \mathcal{C}_j(t),
\end{multline}
with the remainder term of order $\mathcal{C}_j(t) = \mathcal{O}(t^p)$.
Due to the order condition \eqref{eq:p_order_condition}, we conclude that all terms of lower order will eventually cancel out.
Regarding the remainder term, we distinguish between two cases: \\
Case $j \le s$:
\begin{equation}
\begin{split}
&\mathcal{C}_j(t) \\
&= \e^{-i t \ad_{A_s}} \cdots \e^{-i t \ad_{A_{j+1}}} \int_0^t \ud \tau_1 \int_0^{\tau_1} \ud\tau_2 \, \frac{\tau_2^{p-2}}{(p-2)!} \\
&\qquad \times \e^{-i (\tau_1 - \tau_2) \ad_{A_j}} \ad_{-i A_j}^{p-1} \ad_{i A_j} (i B_j) \\
&\: + \sum_{k=j+1}^s \sum_{\substack{q_j + \dots + q_k = p-1 \\ q_k \neq 0}} \e^{-i t A_s} \cdots \e^{-i t A_{k+1}} \\
&\qquad \times \int_0^t \ud\tau \frac{\tau^{q_k-1}}{(q_k-1)! } \frac{t^{q_{k-1} + \dots + q_{j+1} + (q_j + 1)}}{q_{k-1}! \cdots q_{j+1}! (q_j + 1)!} \\
&\qquad \qquad \times \e^{-i (t - \tau) \ad_{A_k}} \ad_{-i A_k}^{q_k} \cdots \ad_{-i A_j}^{q_j} \ad_{i A_j} (i B_j) .
\end{split}
\end{equation}
The expressions with $(q_j + 1)$ (instead of $q_j$) result from the integration in \eqref{eq:abstract_order_series_expansion}, and we use the convention that a summation over an empty range, like $\sum_{k=j+1}^s \cdots$ for $j = s$, evaluates to $0$. \\
Case $j > s$:
\begin{equation}
\begin{split}
&\mathcal{C}_j(t) \\
&= \e^{i t \ad_{A_{s+1}}} \cdots \e^{i t \ad_{A_{j-1}}} \int_0^t \ud \tau_1 \int_0^{\tau_1} \ud\tau_2 \, \frac{\tau_2^{p-2}}{(p-2)!} \\
&\qquad \times \e^{i (\tau_1 - \tau_2) \ad_{A_j}} \ad_{i A_j}^{p-1} \ad_{i A_j} (i B_j) \\
&\: + \sum_{k=s+1}^{j-1} \sum_{\substack{q_k + \dots + q_j = p-1 \\ q_k \neq 0}} \e^{i t A_{s+1}} \cdots \e^{i t A_{k-1}} \\
&\qquad \times \int_0^t \ud\tau \frac{\tau^{q_k-1}}{(q_k-1)! } \frac{t^{q_{k+1} + \dots + q_{j-1} + (q_j + 1)}}{q_{k+1}! \cdots q_{j-1}! (q_j + 1)!} \\
&\qquad \qquad \times \e^{i (t - \tau) \ad_{A_k}} \ad_{i A_k}^{q_k} \cdots \ad_{i A_j}^{q_j} \ad_{i A_j} (i B_j) .
\end{split}
\end{equation}
To assemble everything into a final error estimate, we use the triangle inequality and the fact that the spectral norm of the matrix exponentials evaluates to $1$ (since they are unitary maps).
The inner integration w.r.t.\ $\tau$ can then be performed analytically, and one arrives at
\begin{equation}
\begin{split}
&\norm{\mathcal{T}(t)} %
 \le \sum_{j=2}^K \norm*{\mathcal{C}_j(t)} \\
&\le \sum_{j=2}^s \Bigg( \int_0^t \ud \tau_1 \int_0^{\tau_1} \ud\tau_2 \, \frac{\tau_2^{p-2}}{(p-2)!} \norm*{\ad_{A_j}^p B_j } \\
&\quad + \sum_{k=j+1}^s \sum_{\substack{q_j + \dots + q_k = p-1 \\ q_k \neq 0}} \frac{t^p}{p!} \binom{p}{q_j+1, q_{j+1}, \dots, q_k} \\
&\qquad \qquad \times \norm*{\ad_{A_k}^{q_k} \cdots \ad_{A_j}^{q_j} \ad_{A_j} B_j} \Bigg) \\
&\: + \sum_{j=s+1}^K \Bigg( \int_0^t \ud \tau_1 \int_0^{\tau_1} \ud\tau_2 \, \frac{\tau_2^{p-2}}{(p-2)!} \norm*{\ad_{A_j}^p B_j } \\
&\quad + \sum_{k=s+1}^{j-1} \sum_{\substack{q_k + \dots + q_j = p-1 \\ q_k \neq 0}} \frac{t^p}{p!} \binom{p}{q_k, \dots, q_{j-1}, q_j+1} \\
&\qquad \qquad \times \norm*{\ad_{A_k}^{q_k} \cdots \ad_{A_j}^{q_j} \ad_{A_j} B_j} \Bigg) \\
&= \frac{t^p}{p!} \Bigg( \sum_{j=2}^s \sum_{\substack{q_j + \dots + q_s = p \\ q_j \neq 0}} \binom{p}{q_j, \dots, q_s} \norm*{\ad_{A_s}^{q_s} \cdots \ad_{A_j}^{q_j} B_j} \\
&\: + \sum_{j=s+1}^K \sum_{\substack{q_{s+1} + \dots \\ + q_j = p \\ q_j \neq 0}} \binom{p}{q_{s+1}, \dots, q_j} \norm*{\ad_{A_{s+1}}^{q_{s+1}} \cdots \ad_{A_j}^{q_j} B_j} \Bigg).
\end{split}
\end{equation}
Now applying \cite[Lemma~A.1]{Childs2021} (variation of parameters) to Eq.~\eqref{eq:S_diff_left_right} results in
\begin{equation}
\label{eq:S_diff_expitH}
\mathscr{S}_p(t) = \e^{-i t H} + \int_0^t \ud\tau \e^{-i (t - \tau) H} \mathscr{S}_{\text{left}}(\tau) \, \mathcal{T}(\tau) \, \mathscr{S}_{\text{right}}(\tau).
\end{equation}
Inserting the above bound for $\norm{\mathcal{T}(t)}$ leads to the final result
\begin{equation}
\begin{split}
&\norm*{\mathscr{S}_p(t) - \e^{-i t H}} \le \int_0^t \ud\tau \norm{\mathcal{T}(\tau)} \\
&\le \frac{t^{p+1}}{(p+1)!} \Bigg( \sum_{j=2}^s \sum_{\substack{q_j + \dots + q_s = p \\ q_j \neq 0}} \binom{p}{q_j, \dots, q_s} \norm*{\ad_{A_s}^{q_s} \cdots \ad_{A_j}^{q_j} B_j} \\
&\: + \sum_{j=s+1}^K \sum_{\substack{q_{s+1} + \dots \\ + q_j = p \\ q_j \neq 0}} \binom{p}{q_{s+1}, \dots, q_j} \norm*{\ad_{A_{s+1}}^{q_{s+1}} \cdots \ad_{A_j}^{q_j} B_j} \Bigg).
\end{split}
\end{equation}

\section{Prefactor comparison for the second-order Suzuki formula}
\label{sec:prefactor_comparison}

To demonstrate the improvement offered by Proposition~\ref{prop:tight_bound_strang} compared to Theorem~\ref{thm:higher_order_bounds_prefactors}, we consider a Hamiltonian with three terms, $H = H_1 + H_2 + H_3$.
Evaluating Eq.~\eqref{eq:tight_bound_strang}, expanding the commutators and using the triangle inequality leads to
\begin{equation}
\label{eq:tight_bound_strang_three_terms}
\begin{split}
&\norm*{\mathscr{S}_2(t) - \e^{-i t H}} \\
&\le t^3 \bigg( \frac{1}{24} \norm*{\comm*{H_1, \comm*{H_2, H_1}}} + \frac{1}{12} \norm*{\comm*{H_2, \comm*{H_2, H_1}}} \\
&\hspace{6mm} + \frac{1}{12} \norm*{\comm*{H_3, \comm*{H_2, H_1}}} + \frac{1}{24} \norm*{\comm*{H_1, \comm*{H_3, H_1}}} \\
&\hspace{6mm} + \frac{1}{12} \norm*{\comm*{H_2, \comm*{H_3, H_1}}} + \frac{1}{12} \norm*{\comm*{H_3, \comm*{H_3, H_1}}} \\
&\hspace{6mm} + \frac{1}{24} \norm*{\comm*{H_2, \comm*{H_3, H_2}}} + \frac{1}{12} \norm*{\comm*{H_3, \comm*{H_3, H_2}}} \bigg).
\end{split}
\end{equation}
For comparison, we now apply Theorem~\ref{thm:higher_order_bounds_prefactors} in the same setting.
Using the notation of Eq.~\eqref{eq:product_formula_sequential}, we identify $K = 5$, $A_1 = \frac{1}{2} H_1$, $A_2 = \frac{1}{2} H_2$, $A_3 = H_3$, $A_4 = \frac{1}{2} H_2$, $A_5 = \frac{1}{2} H_1$.
Inserting into Eq.~\eqref{eq:product_formula_error_bound} for $s = 3$ leads to
\begin{equation}
\begin{split}
&\norm*{\mathscr{S}_2(t) - \e^{-i t H}} \\
&\le \frac{t^3}{6} \Bigg( \sum_{\substack{q_2 + q_3 = 2 \\ q_2 \neq 0}} \binom{2}{q_2, q_3} \norm*{\ad_{A_3}^{q_3} \ad_{A_2}^{q_2} \big(\tfrac{1}{2} H_1\big)} \\
&\hspace{6mm} + \sum_{q_3=2} \binom{2}{q_3} \norm*{\ad_{A_3}^{q_3} \big(\tfrac{1}{2} H_1 + \tfrac{1}{2} H_2\big)} \\
&\hspace{6mm} + \sum_{q_4=2} \binom{2}{q_4} \norm*{\ad_{A_4}^{q_4} \big(\tfrac{1}{2} H_1 + \tfrac{1}{2} H_2 + H_3\big)} \\
&\hspace{6mm} + \sum_{\substack{q_4 + q_5 = 2 \\ q_5 \neq 0}} \binom{2}{q_4, q_5} \norm*{\ad_{A_4}^{q_4} \ad_{A_5}^{q_5} \big(\tfrac{1}{2} H_1 + H_2 + H_3\big)}\Bigg) \\
&\le t^3 \bigg( \frac{1}{24} \norm*{\comm*{H_1, \comm*{H_2, H_1}}} + \frac{1}{8} \norm*{\comm*{H_2, \comm*{H_2, H_1}}} \\
&\hspace{6mm} + \frac{1}{12} \norm*{\comm*{H_3, \comm*{H_2, H_1}}} + \frac{1}{24} \norm*{\comm*{H_1, \comm*{H_3, H_1}}} \\
&\hspace{6mm} + \frac{1}{12} \norm*{\comm*{H_2, \comm*{H_3, H_1}}} + \frac{1}{12} \norm*{\comm*{H_3, \comm*{H_3, H_1}}} \\
&\hspace{6mm} + \frac{1}{24} \norm*{\comm*{H_2, \comm*{H_3, H_2}}} + \frac{1}{12} \norm*{\comm*{H_3, \comm*{H_3, H_2}}} \bigg).
\end{split}
\end{equation}
This expression differs from Eq.~\eqref{eq:tight_bound_strang_three_terms} by the prefactor $\frac{1}{8}$ compared to the prefactor $\frac{1}{12}$ in front of $\norm{\comm{H_2, \comm{H_2, H_1}}}$.

\newpage

\section{Commutation relations for the Hamiltonian terms of the Fermi-Hubbard model}
\label{sec:commutation_term_level}

We verify the commutation relations stated in Sect.~\ref{sec:commutators_elementary}.

The relations in Eq.~\eqref{eq:comm_disjoint_support} are clear when noting that the hopping and number operators consist of an even number of creation and annihilation operators, and hence they commute in case they have non-overlapping support.

The statement that number operators always commute, Eq.~\eqref{eq:comm_number_op}, follows from the fact that number operators are diagonal matrices with respect to the standard basis.

For notational conciseness, we will omit the spin index in the following without loss of generality, assuming that all operators share the same spin.

\begin{lemma}
\label{lemma:basic_fermionic_commutator}
If $i, j, k \in \Lambda$ are pairwise different or $i \neq j = k$ or $i = j \neq k$, then
\begin{equation}
\comm{a_i^\dagger a_j^{}, a_j^\dagger a_k^{}} = a_i^\dagger a_k^{}.
\end{equation}
\end{lemma}
\begin{proof}
First consider the case where $i,j,k$ are pairwise different:
\begin{equation}
\begin{split}
\comm{a_i^\dagger a_j^{}, a_j^\dagger a_k^{}} %
&= a_i^\dagger a_j^{} a_j^\dagger a_k^{} - a_j^\dagger a_k^{} a_i^\dagger a_j^{} \\
&= a_i^\dagger a_k^{} a_j^{} a_j^\dagger + a_i^\dagger a_k^{} a_j^\dagger a_j^{} \\
&= a_i^\dagger a_k^{} \left(a_j^{} a_j^\dagger + a_j^\dagger a_j^{}\right) \\
&= a_i^\dagger a_k^{}.
\end{split}
\end{equation}

Next, consider $i \neq j = k$:
\begin{equation}
\begin{split}
\comm{a_i^\dagger a_j^{}, a_j^\dagger a_j^{}} %
&= a_i^\dagger a_j^{} a_j^\dagger a_j^{} - a_j^\dagger a_j^{} a_i^\dagger a_j^{} \\
&= a_i^\dagger a_j^{} a_j^\dagger a_j^{} \\
&= a_i^\dagger a_j^{}.
\end{split}
\end{equation}
Similar calculations complete the proof for $i = j \neq k$.
\end{proof}

Now we consider $\comm{h_{ij}^{}, h_{jk}^{}}$ for $i, j, k$ pairwise different:
\begin{equation}
\begin{split}
\comm{h_{ij}^{}, h_{jk}^{}} %
&= \comm{a_i^\dagger a_j^{} + a_j^\dagger a_i^{}, a_j^\dagger a_k^{} + a_k^\dagger a_j^{}} \\
&= \comm{a_i^\dagger a_j^{}, a_j^\dagger a_k^{}} + \comm{a_i^\dagger a_j^{}, a_k^\dagger a_j^{}} \\
&\: + \comm{a_j^\dagger a_i^{}, a_j^\dagger a_k^{}} + \comm{a_j^\dagger a_i^{}, a_k^\dagger a_j^{}} \\
&= \comm{a_i^\dagger a_j^{}, a_j^\dagger a_k^{}} - \comm{a_k^\dagger a_j^{}, a_j^\dagger a_i^{}} \\
&= a_i^\dagger a_k^{} - a_k^\dagger a_i^{} \\
&= \tilde{h}_{ik}^{}
\end{split}
\end{equation}
where Lemma \ref{lemma:basic_fermionic_commutator} has been applied in the penultimate step.
Notably, the commutator of two adjacent hopping terms has no support on the overlapping mode (here $j$) and is, up to a sign change, equal to a hopping term itself.
In case $i = k$, one observes that $\comm{h_{ij}^{}, h_{ji}^{}} = \comm{h_{ij}^{}, h_{ij}^{}} = 0$, which agrees with $\tilde{h}_{ii}$.
Taken together, we have verified Eq.~\eqref{eq:comm_h_h}.

The commutators involving signed hopping terms follow a similar pattern: for $i, j, k$ pairwise different,
\begin{equation}
\begin{split}
\comm{\tilde{h}_{ij}^{}, \tilde{h}_{jk}^{}} %
&= \comm{a_i^\dagger a_j^{} - a_j^\dagger a_i^{}, a_j^\dagger a_k^{} - a_k^\dagger a_j^{}} \\
&= \comm{a_i^\dagger a_j^{}, a_j^\dagger a_k^{}} - \comm{a_k^\dagger a_j^{}, a_j^\dagger a_i^{}} \\
&= a_i^\dagger a_k^{} - a_k^\dagger a_i^{} \\
&= \tilde{h}_{ik}^{}
\end{split}
\end{equation}
and
\begin{equation}
\begin{split}
\comm{h_{ij}^{}, \tilde{h}_{jk}^{}} %
&= \comm{a_i^\dagger a_j^{} + a_j^\dagger a_i^{}, a_j^\dagger a_k^{} - a_k^\dagger a_j^{}} \\
&= \comm{a_i^\dagger a_j^{}, a_j^\dagger a_k^{}} + \comm{a_k^\dagger a_j^{}, a_j^\dagger a_i^{}} \\
&= a_i^\dagger a_k^{} + a_k^\dagger a_i^{} \\
&= h_{ik}^{},
\end{split}
\end{equation}
again using Lemma~\ref{lemma:basic_fermionic_commutator}.
In case $i = k \neq j$, we obtain $\comm{\tilde{h}_{ij}^{}, \tilde{h}_{ji}^{}} = - \comm{\tilde{h}_{ij}^{}, \tilde{h}_{ij}^{}} = 0 = \tilde{h}_{ii}^{}$, which completes the derivation of Eq.~\eqref{eq:comm_g_g}.
Moreover,
\begin{equation}
\begin{split}
\comm{h_{ij}^{}, \tilde{h}_{ji}^{}} %
&= \comm{a_i^\dagger a_j^{} + a_j^\dagger a_i^{}, a_j^\dagger a_i^{} - a_i^\dagger a_j^{}} \\
&= \comm{a_i^\dagger a_j^{}, a_j^\dagger a_i^{}} - \comm{a_j^\dagger a_i^{}, a_i^\dagger a_j^{}} \\
&= 2 \comm{a_i^\dagger a_j^{}, a_j^\dagger a_i^{}} \\
&= 2 \left(a_i^\dagger a_j^{} a_j^\dagger a_i^{} - a_j^\dagger a_i^{} a_i^\dagger a_j^{}\right) \\
&= 2 \left(n_i^{} (1 - n_j^{}) - n_j^{} (1 - n_i^{})\right) \\
&= 2 \left(n_i^{} - n_j^{}\right),
\end{split}
\end{equation}
finalizing the proof of Eq.~\eqref{eq:comm_h_g}.

As last step, we verify Eq.~\eqref{eq:comm_h_n} by
\begin{equation}
\begin{split}
\comm{h_{ij}^{}, n_j^{}} %
&= \comm{a_i^\dagger a_j^{}, a_j^\dagger a_j^{}} - \comm{a_j^\dagger a_j^{}, a_j^\dagger a_i^{}} \\
&= a_i^\dagger a_j^{} - a_j^\dagger a_i^{} \\
&= \tilde{h}_{ij}^{}
\end{split}
\end{equation}
where we have used Lemma~\ref{lemma:basic_fermionic_commutator} to simplify the commutators.
Eq.~\eqref{eq:comm_g_n} follows analogously.

\end{document}